\documentclass[12pt]{article}

\usepackage{graphicx,float,amsmath,amsthm,cite}

\theoremstyle{theorem}
\newtheorem{prop}{Proposition}
\theoremstyle{definition}
\newtheorem*{definicja}{Definition}

\renewcommand{\theequation}{\arabic{section}.\arabic{equation}}
\newcommand{\beq}[1]{\begin{equation}\label{#1}}
\newcommand{\eeq}{\end{equation}}
\newcommand{\pder}[2]{\frac{\partial #1}{\partial #2}}

\title{SPECTRA AND BIFURCATIONS}
\author{P.~Grochowski$^{a}$ \and W.~Kaniowski$^{b}$
\and B.~Mielnik$^{c}$}

\date{}

\begin{document}
\maketitle
\begin{center}
\textit{$^{a}$Interdisciplinary Centre for Mathematical
and Computational Modelling,\\
Warsaw University,\\
Pawi\'nskiego 5a, 02-106 Warsaw, Poland.\\[0.3cm]
$^{b}$Department of Biophysics,
Warsaw University,\\
\.Zwirki i Wigury 93, 02-089 Warsaw, Poland.\\[0.3cm]
$^{c}$Departamento de Fisica, CINVESTAV,\\
A.P.~14-740, 07000 M\'{e}xico DF, M\'exico.}
\end{center}

\begin{abstract}
The concept of spectrum for a class of  
non-linear wave equations is studied.   Instead  of  looking 
for stability, the key to the spectral structure is found 
in the instability phenomena (bifurcations). This  aspect 
is best seen in the `classical model' of  the  non-linear 
wave mechanics. The solitons (macro-localizations) are a 
part  of   the   non-linear   spectral   problem;   their 
bifurcations reflect the dynamical symmetry breaking.
The computer simulations suggest that the bifurcations of the
asymptotic behaviour occur also for the general, non-stationary
states. A~phenomenon of the soliton splitting is observed.
\end{abstract}

\section{Introduction}
\setcounter{equation}{0}

     In the last decades one can observe a renewed interest in 
non-linear wave equations \cite{Haa,Kib,IBB,Wei,Cza0,Cza1,Do1,Do2,Dodo,Die,Gra}.
The simplest class of non-linear  
Schr\"{o}dinger's equations creates a temptation to formulate 
a ``quantum mechanics" with 
$\psi \in L^{2}(R^{3})$, $|\psi|^{2}$~as the 
probability density \cite{IBB,Cza1,Mie} but with the superposition 
principle broken. In turn, the non-linear Schr\"{o}dinger's equation 
with an atypical kinetic energy \cite{Mie,Aro}:
\beq{i1}
i\pder{\psi}{t} = -\frac{1}{2} \Delta (|\psi|^{\alpha} \psi) 
+V({\bf x},t)\psi
\eeq
$(\alpha \in R)$, has the absolutely conservative integral: 
\beq{i2}
N[\psi] = \int_{R^{3}}^{}|\psi|^{2+2\alpha}d_{3}x = {\rm const}
\eeq
suggesting a statistical theory with $\psi \in L^{p}(R^{3}),\;
p=2+2\alpha$ and  with $|\psi|^{p}$ defining the probability density. 
A slightly different equation: 
\beq{i3}
i\pder{\psi}{t} = -\frac{1}{2} |\psi|^{-\alpha}
\Delta (|\psi|^{\alpha} \psi) +V({\bf x},t)\psi
\eeq
admits the conservative integral: 
\beq{i4}
N[\psi] = \int_{R^{3}}^{}|\psi|^{2+4\alpha}d_{3}x = {\rm const}
\eeq
asking for the space of states $L^{k}(R^{3})$ $(k=2+4\alpha$).  
Curiously, (\ref{i3}) has the 
same  spectrum  as  the  conventional  Schr\"{o}dinger  equation
(a counterexample against the belief that the non-linearity can be
tested by observing the spectral frequencies). 
The  eqs. (\ref{i1}-3) are not Galileo covariant; 
however, cases of Galileo invariant wave mechanics  were 
recently found  by  Doebner and Goldin \cite{Do1,Do2}; see also   
Dodonov and  Mizrahi \cite{Dodo}, and Natterman \cite{Nat}.  
One  of hopes in the
non-linear schemes is that they might help to understand the collapse 
of the wave packets (see e.g.~Gisin \cite{Gis}),
 but one of obstacles is that the  
non-linearity  could generate faster than light signals in a  
non-linear  analogue  of  EPR arrangements (a disquieting observation 
of  Gisin  \cite{Gisin}  and  Czachor \cite{Cza0} leaves the 
`fundamental non-linearity' in defense, but  not in defeat!).

     Apart of fundamental reasons, the non-linear wave eqs. 
might be of practical interest,  as  tools  to  describe  dense 
clouds  of 
interacting   quanta.   To   this  subject   belong   all   
variants   of `self-consistent' wave  mechanics \cite{Self,Cap}, 
the theories which model the feedback interactions of a micro-object  
with  a mezoscopic or  macroscopic  medium, in molecular \cite{Ber,Les} or solid 
state  physics  \cite{Die,Gra}.  In  all  
these  schemes,  the `localizations' (bound states) bring a relevant 
information. However, some structural problems are still open.

     One of unsolved questions in non-linear  theories  is  the  
problem  of spectrum. Is it pertinent to  define  spectra  for  
non-linear operators \cite{Bra,IBB,Cza1}?  Looking  for  
similarities  between  the   linear   and non-linear cases, one 
might be tempted by  the  idea  of  stability.  In fact, in the 
orthodox quantum theory the bound states are stationary and stable; 
in non-linear case the same concerns solitons; a lot of  authors 
marvel about the soliton ability to survive collisions!  Yet, the  
story has its opposite side,  which  seems  to be as relevant for  the  
non-linear spectral problem. 

      One of curious aspects of the Schr\"{o}dinger's eigenvalue 
equation in 1-space  dimension  is  the  possibility  of 
reinterpreting  the  space coordinate $x$ as the ``time" ($x=t$), 
the wave function $\psi$ as the  coordinate 
and the derivative $\psi '$ as the  momentum  of  a  certain  
classical  point particle \cite{Pru,Man,Col,San,book,Mre}. 
What one  obtains  is  a  classical 
model for  quantum phenomena or vice versa \cite{Col,Wo1,Wo2} (see, e.g.,
the interpretation of the Saturn  rings  as   spectral 
bands \cite{Avr}). 
It turns out that the `classical
image' throws also some new light onto the non-linear spectral  
problem. 
It shows that  the  eigenstates  correspond  to  bifurcations 
\cite{Bra,Rab,Eas,Tur,Sch,Hef}; in a sense, they are ``born of 
instability"! 

Our paper is precisely dedicated to the instability (bifurcation) 
aspects of the one dimensional spectral problem. We shall show that they provide the 
most natural bridge between the linear and non-linear cases, leading also to the 
easiest numerical algorithm to determine 
the spectral values. Among many  models  which can be used to illustrate this, 
we have selected the simplest  one,  
with  some hope that observations presented below might turn 
generally useful. 

\section{The `classical portrait'.}
\setcounter{equation}{0}

We shall consider the non-linear Schr\"{o}dinger's equation in 
1-space dimension: 
\beq{p1}
i\frac{d}{dt}\Psi(x,t) =
-\frac{1}{2} \frac{d^{2}\Psi}{dx^{2}} + V(x)\Psi + 
\varepsilon f(|\Psi|^{2})\Psi
\eeq
where $V(x)$ is an external potential and $f()$ a given function defining 
the non-linearity.  The  stationary  solutions  $\Psi (x,t)$ = 
exp $(-i E t ) \psi(x)$ then fulfill:
\beq{p2}
-\frac{1}{2} \frac{d^{2}\psi}{dx^{2}} + [V(x)-E]\psi + \varepsilon f(|\psi|^{2})\psi =0
\eeq
The simple harted  analogue  of  an  eigenstate  can  be  introduced 
without difficulty \cite{Bra,Rab,Eas,Tur,Sch,Hef}.  
Whenever  $\psi$ in (\ref{p2}) is localized, i.e. vanishes for 
$x \rightarrow \pm \infty$, then $\psi$ will be called  a 
`localized state' and $E$ will be interpreted as a discrete frequency 
eigenvalue of (\ref{p2}). We use the traditional symbol $E$ but we speak 
about the frequency instead of energy eigenvalue to remind that for the 
non-linear Schr\"{o}dinger's equation (\ref{p2}) the parameter $E$ 
may have no energy interpretation (this point is widely discussed in 
\cite{IBB,Cza1}). 

Denoting now $x=t$, $\psi=q_{1}+iq_{2}$, $\psi^{\prime}=p_{1}+ip_{2}$ 
one immediately reduces 
(\ref{p2}) to the Newton's equation of motion: 
\beq{p3}
\frac{d{\bf q}}{dt} = {\bf p},~~~
\frac{d{\bf p}}{dt} = 2[V(t)-E]{\bf q} + 2\varepsilon f({\bf q}^{2}){\bf q}.
\eeq
for a classical point particle  in  a  radial,  centrally  symmetric 
force field in 2 space dimensions, where ${\bf q,p}$ denote the 
two-component position and momentum vectors. The Hamiltonian is: 
\beq{p4}
H(t) = {\bf p}^{2}/2 + [E-V(t)]{\bf q}^{2} - \varepsilon F({\bf q}^{2}),
\eeq
with $F(\zeta)= \int_{}^{} f(\zeta)d\zeta$.
For every $E \in R$ the eq.(\ref{p3}), of course, has a  family 
of solutions labelled by 2 complex (or 4 real) parameters,  but very 
seldom it has solutions vanishing for both 
$t \rightarrow + \infty$ and $t \rightarrow - \infty$. More seldom 
even they will vanish quickly enough to assure: 
\beq{p5}
\int_{-\infty}^{+\infty} |\psi (x)|^{2}dx =
\int_{-\infty}^{+\infty} {\bf q}(t)^{2}dt < + \infty
\eeq
Whenever (\ref{p3}) admits non-trivial solutions vanishing at 
$t \rightarrow \pm \infty$, $E$
is an eigenvalue (proper frequency) of (\ref{p2}). We adopted
again the traditional concept of the spectrum in the 
non-linear case, \cite{Bra,Rab,Eas,Tur,Sch,Hef}, 
to be further discussed in our section 5.

\section{Classical orbits and bound states} 
\setcounter{equation}{0}

    One of advantages of the `classical picture' is that it  
permits  to exploit experiences  of  classical  mechanics  to  
describe  the  `bound 
states' (\ref{p3}) and one of the most obvious ideas is 
to  use  the  classical  motion integrals. As the  force  field  in  
(\ref{p3}-4)  is  radial  and  centrally 
symmetric,  the  angular  momentum  of  each  trajectory  is constant. 
Introducing the polar variables $q_{1}$ = $r \cos\alpha$, 
$q_{2}$ = $r \sin\alpha$, one has: 
\beq{c1}
M = {q_{1}}{p_{2}}-{q_{2}}{p_{1}} = r^{2}\dot{\alpha} = \text{const}.
\eeq     
The canonical eqs. (\ref{p3}) now imply: 
\beq{c2}
\frac{d^{2}r}{dt^{2}} = 2[V(t)-E]r + 2\varepsilon f(r^{2})r 
+ M^{2}/r^{3}
\eeq
interpretable as Newton eq. of motion in 1  space  dimension. 
If $M \not = 0$, the last term prevents $r$ (and ${\bf q}$) from 
tending to zero: 

\begin{prop}\label{prop1}
If $V(t)$ is limited from below 
and $M \not = 0$ then the solution $r(t)$ of (\ref{c2}) 
cannot tend to zero neither for a finite $t$ nor for 
$t \rightarrow \pm \infty$. 
\end{prop}
\begin{proof}
Indeed, for a non-trivial trajectory 
($r(t) \not \equiv 0$) two first terms on the right side of 
(\ref{c2}) create either repulsive or limited attractive forces and 
cannot counterbalance  the third term $M^{2}/r^{3}$ which prevents the 
material point from approaching  too close to zero. 
\end{proof}

    As a consequence, for $M \neq 0$ the integral (\ref{p5}) diverges for any 
$E \in R$. Thus, the non-trivial localized solutions of (\ref{p3}), 
(with $q \not \equiv 0$) if they exist, must fulfill 
$M$=$0$ $\Rightarrow \dot{\alpha}$=$0$  $\Rightarrow$  $\alpha$=$const$.
Without loosing generality, all bound states of (\ref{p2}) can be 
therefore obtained  for $\alpha \equiv 0$, i.e., for $\psi$ real. 
In terms of the trajectory  interpretation  (\ref{p3}) 
it means that the bound states can be determined  just  by  
solving  the 1-dimensional (instead of the 2-dimensional) motion 
problem with:
\beq{c3}
H(q,p,t) = p^{2}/2 + [E-V(t)] q^{2}  - \varepsilon F(q^{2}),
~~~q, p \in R
\eeq     
and
\beq{c4}
\frac{dq}{dt} = p,~~~
\frac{dp}{dt} = 2[V(t)-E]q + 2\varepsilon f(q^{2})q.
\eeq

\section{ Spectra as bifurcations.} 
\setcounter{equation}{0}

While values of $E$ which permit trajectories (\ref{c4})
vanishing on both ends 
$t \rightarrow \pm \infty$ are exceptions, yet for any $E$ 
eq. (\ref{c4}) typically admits a subclass of solutions 
vanishing for $t \rightarrow - \infty$ (the `left vanishing cues'), 
as well as another subclass vanishing for $t \rightarrow + \infty$ 
(the `right vanishing cues'). Indeed:

\begin{prop}\label{prop2}
Suppose $f(\zeta)$ is continuous  
in  $[0,+\infty)$ with  $f(0)=0$, while $V(t)$ is  
defined and continuous outside of a finite interval 
$[a,b]$ with two (proper or improper) limits:
\beq{s1}
V_{1} =  \lim_{t \rightarrow - \infty} V(t) > - \infty,~~~	 
V_{2} =  \lim_{t \rightarrow + \infty} V(t) > - \infty  
\eeq     
Then for any $E<V_{2}$, $t_{o}>b$, there is $K_{2} > 0$ such
that for any $q_{o}\in R$, $|q_{o}| \leq K_{2}$ the  Hamiltonian  
(\ref{c3})  admits  at 
least one  trajectory  $q(t)$  with $q(t_{o})$ = $q_{o}$  and  $q(t) 
\rightarrow 0$  for  $t \rightarrow +\infty$. 
Similarly, for any $E<V_{1}$, $t_{o}<a$, there is $K_{1} > 0$ such
that for any $q_{o} \in R$, $|q_{o}| < K_{1}$, there  is  at  
least  one  integral 
trajectory $q(t)$ with $q(t) \rightarrow 0$ for $t \rightarrow - \infty$.
\end{prop}
\noindent({\bf Proof}  is  an  exercise  in 
shooting \cite{book,Mre}. {\it Note:} the conclusion of the theorem 
holds also  for 
some non-linear eqs. with $f(\zeta)$ singular at $\zeta$=0 \cite{IBB}). 

To illustrate the structure of the cues 
we shall consider the case of a finite potential
well with $V(t) \leq 0$ and $V(t) \equiv 0$ 
outside of a 
finite interval  $[a,b]$.  The  trajectory  q(t)  outside  of  [a,b]  then 
corresponds to the motion of a classical point in  a  static  potential, 
with the Hamiltonian: 
\beq{s2}
H(t) = H_{o} = p^{2}/2 + E q^{2}  - \varepsilon F(q^{2}),~~~t\not\in[a,b]   
\eeq     
The solutions asymptotically vanishing at $t \rightarrow 
\pm \infty$ are the  orbits  for 
which the Hamiltonian (\ref{s2}) vanishes, i.e: 
\beq{s3}
p = \pm q \sqrt2 \sqrt{-E + \varepsilon F(q^{2})/q^{2}} = \pm \eta(q,E) 
\eeq     
the signs + (or - ) label the solutions vanishing  
at $t \rightarrow - \infty$ 
($t \rightarrow + \infty$), respectively. If one forgets about the 
exact  t-dependence,  the 
condition (\ref{s3}) determines just two  evolution-invariant  
curves (a part of the `phase portrait' of (\ref{p2}), 
cf.\cite{Holms}): 
\beq{s4}
I_{\pm}(E) = \{(q,p): p=\pm \eta(q,E)\}.                 
\eeq     
Under the exclusive  influence  of  the  free  Hamiltonian  
(\ref{s2}) (i.e., for $V(t)\equiv 0$),  the  canonical  evolution  
produces  a `curvilinear squeezing' which distinguishes 
the $I_{\pm}(E)$ lines:  $I_{+}(E)$  expands  (the  points on 
$I_{+}(E)$ escape from the phase space origin as $t$ increases 
from $t $=$ -\infty$), whereas $I_{-}(E)$ 
shrinks (the points of $I_{-}(E)$ tend to the origin as 
$t \rightarrow + \infty$). 
If  one  solves (\ref{s3}) including the exact time dependence: 
\beq{s5}
\dot{q} = \pm \eta(q,E).                 
\eeq   
then the points on $I_{+}(E)$ originate the left vanishing  cues  
(i.e.  the 
solutions of (\ref{p2}-3) vanishing at $t\rightarrow -\infty$), 
whereas the points of $I_{-}(E)$,  the 
right vanishing cues (tending to zero as $t\rightarrow +\infty$).
If $I_{+}(E)$ and  $I_{-}(E)$ don't connect outside of the
origin,   
the solutions in the form of bound states can arise only due to the
potential $V(t) \not\equiv 0$.
The number $E$ is an eigenvalue of (\ref{p2}) if the evolution 
enforced by $V(t)$ in the  interval 
$[a,b]$ transforms one of the left-vanishing cues into one  of  the  
right-vanishing cues (thus drawing a wave function which vanishes  
on the  both ends $t\rightarrow\pm\infty$; 
compare \cite{Mre}). 
The effect has rather little to do  with  the  linearity.  
The linearity  is a marginal property of the orthodox 
Schr\"{o}dinger's dynamics - whereas the existence of the `bound states' is the 
topological phenomenon, caused by an abrupt change (bifurcation)
which produces the homoclinic orbits of (\ref{c4}) (compare Guckenheimer and Holmes \cite{Holms}). The mechanism of this effect can be 
easily monitored.

Suppose, we keep $E$ fixed  but change $V(t)$ putting 
$V(t)=\lambda\phi(t)$, with $\phi$ fixed and $\lambda>0$ 
variable. Observe then a congruence of orbits sticking at 
$t=a$ from a given point of $I_{+}(E)$. If  
$\lambda=0$, then the orbit sticks to $I_{+}(E)$ 
forever. Assume now, we switch on slowly the potential term  
$V(t)=\lambda\phi(t)$.  For 
small $\lambda>0$ the motion is slightly modified 
(the trajectory  is pushed 
out of $I_{+}(E)$) though it returns to $I_{+}(E)$ asymptotically  
as  $V(t)$  disappears. 
When $\lambda$ increases, the type of the motion changes. 
In the  time  interval where $V(t)-E=\lambda\phi(t)-E<0$  the  
Hamiltonian  (\ref{c3})  becomes an attractive 
anharmonic oscillator and causes a circulation around the origin 
instead of squeezing. Everytime the circulating point  $(q(b),p(b))$  
crosses  the line $I_{-}(E)$, the asymptotic behaviour of the 
trajectory suddenly changes. At the exact bifurcation value of 
$\lambda$ the point $(q(b),p(b))$ ends up on $I_{-}(E)$; 
henceforth,  for $t>b$, the potential-free motion (\ref{s2})
returns $(q(t),p(t))$ to zero drawing an exceptional, closed 
orbit (a bound state of $\lambda\phi(t)$ for the eigenvalue  $E$).  
An  analogous 
picture arises for $V(t)$ fixed but $E$ changing. To illustrate  
the phenomenon, we have applied 
the computer simulation to obtain a family 
of canonical trajectories of (\ref{c3}) for a fixed  $E<0$  and  
the potential well $\phi(t)$ given by the forthcoming formula (\ref{10.1}). 
The bifurcations lead to  the bound states illustrated on Fig.~\ref{fig1}. 
\begin{figure}[htb]
\centering
\includegraphics[width=0.6\linewidth]{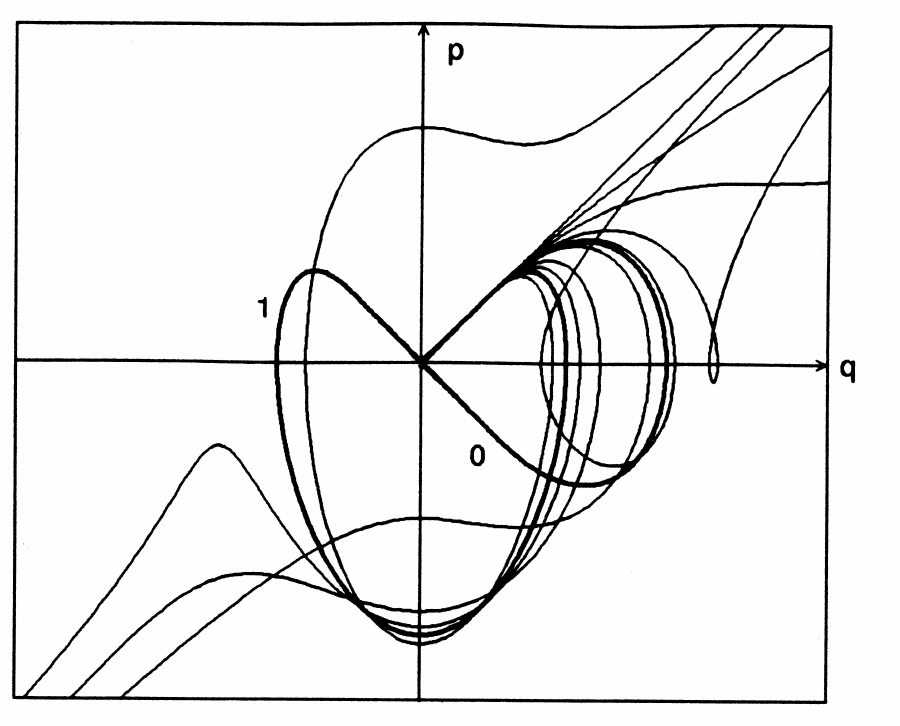}
\caption{A congruence of classical trajectories 
of (3-5) in process of bifurcation.
$V(\cdot)$ given by~\eqref{10.1}. The bold lines 0, 1 mark two 
lowest bound orbits created at the bifurcation values of
$\lambda$~\cite{Mre}. An analogues phenomenon can be seen 
for $\lambda$ fixed and $E$ varying.\label{fig1}}
\end{figure}

Note, that we have arrived at a  certain  general  scheme  
for generating the bound states, which no longer requires linear 
spaces and  linear  operators. A similar phenomenon would occur 
on any 2-dim. symplectic manifold  ${\cal P}$ (compare Klauder~\cite{Klau}), with two flows
generated by two `antagonistic' vector fields A and B sharing a 
common fixpoint 0. 

The flow A should be a {\it squeezing, } with a {\it saddle point}
at 0 and with the phase portrait dominated by two intersecting 
invariant lines:  $I_{+}$ expanding, $I_{-}$ shrinking 
(see Fig.~\ref{fig2}). The field B, in turn, should be a circulation, 
with orbits in form of closed loops surrounding the fixpoint 0 
(compare Guckenhaimer and Holmes \cite[p.52, Fig.~1.8.6]{Holms}).

\begin{figure}[H]
\centering
\includegraphics[width=0.85\linewidth]{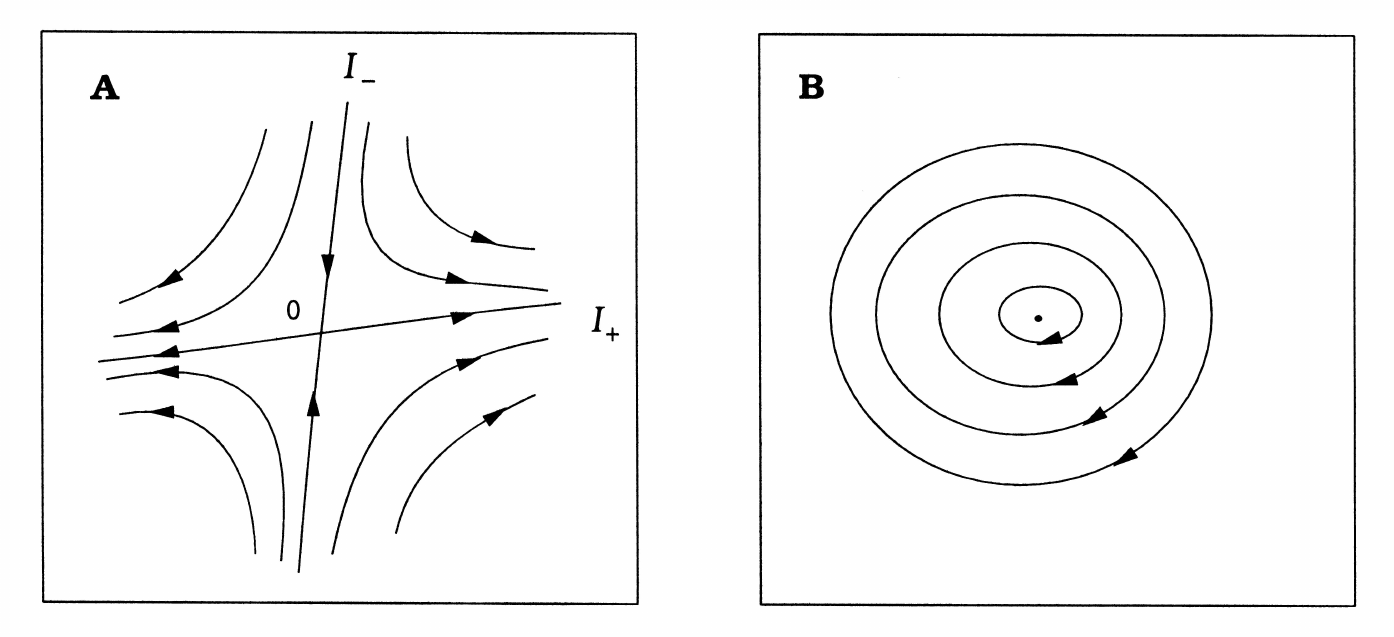}
\caption{The imitation of the spectral phenomenon by two vector
fields, ${\bf A}$ and ${\bf B}$, with coinciding fixpoints
of different types on an arbitrary 2-dim surface. 
The bifurcations of the orbit driven by 
${\bf A}+\lambda \phi(t){\bf B}$ occur whenever the circulation 
generated by ${\bf B}$ changes the assymptotic form of the trajectory
allowed by ${\bf A}$ (compare Guckenheimer and Holms~\cite[p.~52]{Holms}.}
\label{fig2}
\end{figure}

If now the phase point ${\bf q}(t)\in$ ${\cal P}$ moves 
under the influence of a combined vector field ${\bf A} + \lambda 
\phi(t){\bf B}$ (where $\lambda \ge 0$ is a variable amplitude and 
$\phi(t) \ge 0$), then, as $\lambda $ increases, the number of 
intersections of the phase trajectory with the shrinking line 
$I_{-}$ grows too.
At each new intersection, the trajectory bifurcates, forming an 
exceptional, closed orbit, interpretable as a bound state. 

     Our example is oversimplified - but it shows that the mechanism
of creation  
of the  bound  states (at least in the 1-dimensional case) is  not  
metrical  but {\it par  excellence}  topological. 
If the theory is non-linear, the orthogonality of the bound 
states desappears - but the bifurcation mechanism still works,
distinguishing the spectral parameters for the non-linear  system. 
Notice, that 
the idea of spectrum as a sequence of bifurcations has  already  
emerged in some mathematical areas,  e.g. in studies of chaotic   
systems \cite{Bambi1,Bambi2}. It is still an open problem whether 
the similar idea
could work in higher dimensions or for the abstract non-linear 
operators. We shall see, however, that as far as 1-space dimension 
is considered, it leads to some efficient numerical techniques.

\section{ Numerical algorithm and Zakharov well. }
\setcounter{equation}{0}

      To determine numerically the  bifurcation spectra, we 
propose a simple variant of the shooting method involving only the 
integration in a finite interval. 
To illustrate it let's consider again 
eq. (\ref{p2}) with $V(t)$ forming a limited potential well  
($V(t)\equiv 0$ for $t\not\in [a,b]$). 
For any $E\in R$ we then choose an initial point ${\bf q}(a)$
in $I_{+}(E)$ and integrate (\ref{p3}) in $[a,b]$ finding the `final 
point' ${\bf q}(b)=(q(b),p(b))$.

Whenever for an `initial point' 
${\bf q}(a)=(q_a ,p_a)\in I_{+}(E)$ the `final point' ${\bf q}(b)$ 
happens to be on $I_{-}(E)$, the shooting 
exercise was a success: the number E is then an eigenvalue of 
(\ref{p2})  and 
the orbit of (\ref{p3}) [defined by the initial condition  
${\bf q}(a)=(q_a,\eta(q_a , E))$] 
represents a bound state of (\ref{p2}). 
The set of data $(q_a,E)$ for which this 
happens, typically, forms a sequence of lines on the $(q_{a},E)$  
diagramme. The different branches $E_n=E_n(q_{a})$ correspond to
the different numbers of times the phase trajectory crosses
the shrinking line $I_{-}(E)$ for $t\in(a,b)$ (see Fig.~\ref{fig1}). 
The non-trivial $q_{a}$-dependence
of the branches $E_{n}$ reflects the fact that for the 
non-linear eq. (\ref{p2}) the eigenvalues, in general, 
depend on the norm.

    The existence of the localized states (solutions vanishing at 
$t\rightarrow\pm\infty$) does not yet assure that they must be 
square integrable.  In  fact,  for the eq. (\ref{p2}) the cues are 
determined by the  structural function $F(\zeta)$; their square 
integrability depends on  the  convergence of the integrals: 
\beq{n1}
\int_{} \frac{\dot{q}qdt}
{\sqrt{-E +\varepsilon F(q^{2})/q^{2}}} =
\int_{}^{} \frac{d\zeta}{\sqrt{-E +\varepsilon F(\zeta)/\zeta}} 
\eeq     
around $\zeta=q^{2}=0$. 
Presumably, there might be non-linear theories with localized  
solutions vanishing so slowly at $t \rightarrow\pm\infty$ that 
the integrals (\ref{n1}) diverge. The physical sense of such 
localizations is an open problem. Below, we shall consider 
non-linearities for which this does not occur. If this is the case, 
each (non-trivial) bound state possesses a finite norm (\ref{p5}) 
which apports an essential physical information (in contrast to 
the linear theory, where the norm is arbitrary).
The information about the norm, however, is difficult  for 
the  computer  manipulations: (\ref{p5}) is finite only for 
exceptional trajectories (bound states); for all other solutions is
infinite. To facilitate numerical operations we thus introduced 
the {\it pseudo-norm}, well defined, continuous, for all  solutions,  
and coinciding  with (\ref{p5}) whenever the solution is localized. 

\begin{definicja}\label{def1}
For any integral trajectory $q(t)$, 
the {\it pseudonorm} $N(q)$ is: 
\begin{multline}
\label{n2}
N(q) = \int_{-\infty}^{a} q_{+}(t,E)^{2} dt + 
       \int_{a}^{b} q(t)^{2} dt + 
       \int_{b}^{+\infty} q_{-}(t,E)^{2} dt\\ =
       N_{+}(q) +N_{o}(q) +N_{-}(q)
\end{multline}  
where $q_{\pm}(t,E)$ are two `vanishing cues', defined 
in $(-\infty,a]$ and $[b,+\infty)$, 
joining  $q(t)$  at  $t=a$  and  $t=b$   respectively   
(i.e.,   $q_{+}(a,E)=q(a)$; 
$q_{-}(b,E)=q(b)$, {\it without demanding} 
the continuity of the derivatives).  
\end{definicja}
From the definition, (\ref{n2}) is always finite; moreover, 
if $q(t)$ is a bound state, then  $q_{\pm}(t,E)=q(t)$  and  
the  pseudonorm  (\ref{n2})  reduces  to the true norm (\ref{p5}).  
Each  {\it isonorm  line}
$N$=const, typically, crosses the eigenvalue lines $E=E_n(q_{a})$  
($n$=0,1,...); 
the intersections determine the eigenvalues for the  bound  
states of the same norm $N$. In particular, the line $N$=1   
defines the `traditional' sequence of  eigenvalues  
$E_n$  ($n$=0,1,...) for all  bound  states of norm 1. 

To check the method, we have found the `localizations' for  the 
eq. of Zakharov type \cite{Zakh}: 
\beq{n3}
-\frac{1}{2} \frac{d^{2}\psi}{dx^{2}} + [V(x)-E]\psi 
+ \varepsilon|\psi|^{2} \psi = 0
\eeq     
in presence of the rectangular potential well symmetric with 
respect to $t=0$: 
\beq{n4}
V(t)  \left\{ \begin{array}{ll}
              = V_{o} & \mbox{for $t\in [-b,b]$} \\
              \equiv 0 & \mbox{for $t\not\in [-b,b]$}
              \end{array}
	\right.,~~~V_{o} < 0.
\eeq     
Here, $a=-b$, $f(\zeta)=\zeta$, $F(\zeta)=\zeta^{2}/2$, 
and the expanding/shrinking curves $I_{\pm}(E)$  are 
given by: 
\beq{n5}
p(q) = \pm\sqrt{\varepsilon  q^{4} - 2Eq^{2}}.
\eeq
To find the norms, it helps that the bound states  are of 
definite parity inside of $[-b,b]$. In fact, since 
$H_o = p^{2}/2 + [E - V_o]q^{2} - \varepsilon q^{4}/2$  
is a conserved quantity in $[-b,b]$, one has  $
q(b)=\pm q(-b)= \pm q_a$  
for  any  bound  state. Moreover, the Zakharov eq. (\ref{n3}) belongs 
to the list of cases  where the integrals (\ref{n1}-2) are explicitly 
known:
\beq{n6}
N_{\pm} =
\frac{2}{\varepsilon\sqrt{2}}
\left(\gamma\sqrt{|E|+\varepsilon q_{a}^{2}/2}
- \sqrt{|E|}\right),
\eeq
where $\gamma=+1$ for the {\it short cues} and $\gamma=-1$ 
for the {\it long cues} (the last case possible only for 
$\varepsilon < 0$).
The numerical calculus intervenes only in the finite  interval 
$[-b,b]$. We have applied the Runge-Kutta method integrating the 
canonical 
eqs. (\ref{p3}) in $[-b,b]$ for the set of the initial points  
${\bf q}_{a}=(q_{a}, \eta(q_{a},E))\in I_{+}(E)$. 
For each such ${\bf q}_{a}$ we have 
determined the  sequence  of  values $E=E_n$ for which the 
`end points' reach $I_{-}(E_n)$. Differently  than  in  the 
linear theory, the whole ladder depends on the initial value $q_{a}$ 
giving 
rise to the sequence of functions $E=E_n(q_{a})$ as shown on our  
Fig.~\ref{fig3}.  Their intersections with the isonorm lines 
determine the eigenvalues for each given norm. 
\begin{figure}[htb]
\includegraphics[width=1.0\linewidth]{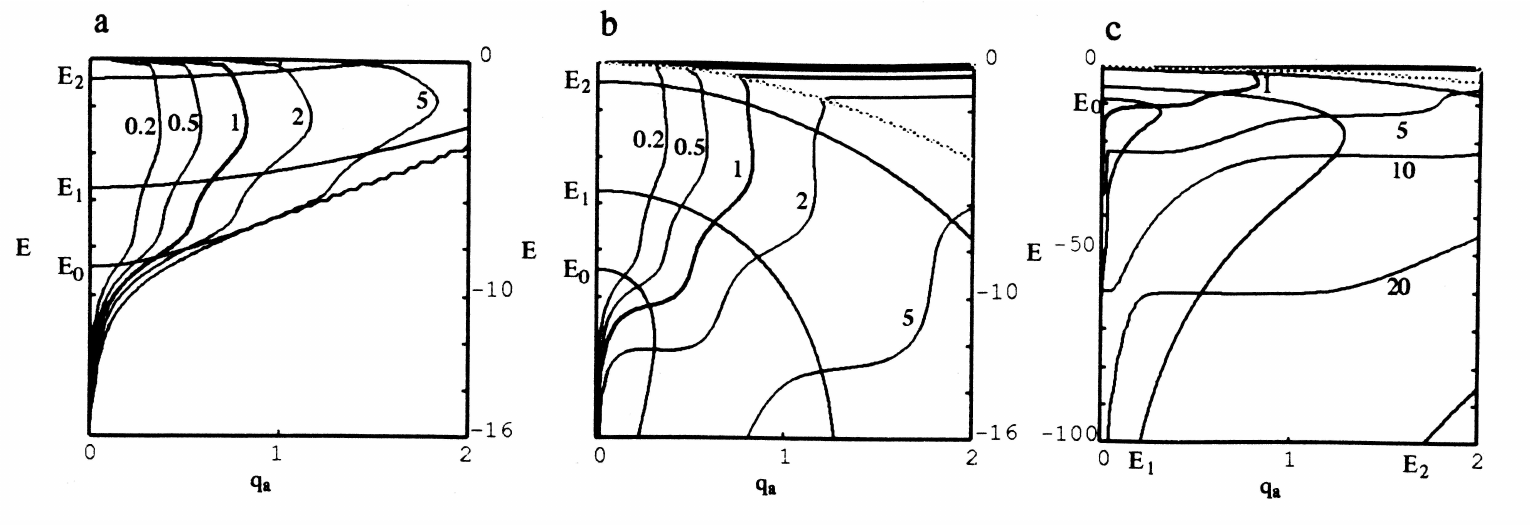}
\caption{
Three maps of the norm dependent $E$-values for the Zakharov packets
in the square potential well with the bottom at $V_{o}$ = -10 and
$b = 1.6$. The fat lines represent the packet with the norm 1. The
eigenvalues for packets of any given norm are determined by the
intersections of the $E_{n}$-lines with the corresponding norm line.
(a) $\varepsilon > 0$ (the self-repulsive packets) the eigenvalues are
above the orthodox ones; (b-c) $\varepsilon < 0$ (self-attractive 
packets) the ground state admits $E_{o}$ below the 
bottom of the well (e.g. $E_{o} \approx -10.4$ for the packet norm 1,
$E_{o} \approx -12.2$ for the packet norm 2). The gray lines fence off
top-right regions of the diagrammes where the pseudonorm 
is indetermined.\label{fig3}}
\end{figure}

As can be observed, the eigenvalues for the states  of  small  norms 
(`little eigenstates') are almost the same as in the linear theory.  
The best conceptual analogy with the orthodox  (linear)  wave  
mechanics  is 
achieved on the isonorm line $N=1$ ( $|q(t)|^{2}$ interpretable simply
as the probability density). For the 
`wave packets' of norms $>1$ the statistical interpretation is no  
longer 
obvious; more appealing would be to interpret them as  localized `drops 
of quantum matter' obeying the non-linear eq. (\ref{p1}-2) 
due to the internal interactions (compare \cite{Die,Gra}). 
Indeed, such an idea is recently adopted in works dedicated to  
boson condensations \cite{Boson1,Boson2,Sack,Dum}.
Observe that for the bound states of high norm the lowest eigenvalue
can be much below the bottom of the well (Fig.~\ref{fig3}c); 
a phenomenon unknown in the linear theory, of tentative interest 
in physics of the condensed mater \cite{Boson1,Boson2,Sack,Dum,Bare}.

    It might be interesting to notice 
that our algorithm works also for the non-linear model  recently 
proposed by Diez et al. \cite{Die,Gra} permitting 
to determine  the  bound  states  in 
case of a double barrier. For $V(t)\equiv 0$ the asymptotic cues 
are exactly  as in the linear theory.

\section{ The Zakharov localizations in a $\delta$-well. } 
\setcounter{equation}{0}

    An extremally simple solution  of  the  spectral  problem  
(\ref{n3})  is obtained for $V(t)$ in form of a $\delta$-well: 
$V(t)=-\Omega\delta(t)$, $\Omega>0$. The  non-linear 
equation (\ref{p2}) for $E<0$ traduces itself into a canonical  
motion  problem 
for a  classical  point  moving  under  the  influence  of  a  constant 
potential  $Eq^{2} - \varepsilon F(q^{2})$ ,  corrected  at  $t=0$  
by  a  sudden 
attractive shock. The time-dependent classical Hamiltonian reads: 
\beq{z1}               
H(t) = p^{2}/2 + [E + \Omega\delta(t)] q^{2}  - \varepsilon F(q^{2})
\eeq
The trajectory ${\bf q}(t)=(q(t),p(t))$  which  vanishes  at  both  
extremes $t\rightarrow\pm\infty$ 
is composed exclusively of two vanishing cues, with the $p$-jump 
at  $t=0$  caused  by  the  $\delta$-pulse  of an attractive force.   
Denote $q_o=q(0_{-})=q(0_{+})$  and  $p_o=p(0_{-})$.  
The  expression  (\ref{s3})   then   requires $p(0_{+})=-p_{o}$. 
On the other hand, the momentum jump is produced by the pulse 
of force: 
\beq{z2}               
\Delta p = p(0_{+})-p(0_{-}) = -2p_{o} =
        \int_{0_{-}}^{0_{+}} {\cal F}(t)dt = 
        -2\Omega \int_{0_{-}}^{0_{+}} q(t)\delta(t)dt = -2\Omega q_{o} 
\eeq
Taking $p_{o}$ from eq. (\ref{s3}) one has: 
\beq{z3}
p_{o} = q_{o}\sqrt2 \sqrt{-E + \varepsilon F(q_{o}^{2})/q_{o}^{2}} = 
\Omega q_{o}
\eeq
and so:
\beq{z4}               
E = - \Omega^{2} /2 + \varepsilon F(q_{o}^{2} )/q_{o}^{2}
\eeq
i.e., the standard eigenvalue $-\Omega^{2}/2$ is corrected 
by the non-linear term. Assuming that $F\geq 0$ in vicinity 
of zero, one sees again that for $\varepsilon<0$ (the self-attractive 
packets) the creation 
of a localized states is possible for a lower $E$, 
whereas for $\varepsilon>0$ 
(self-repulsion) the non-linearity prevents to trap the packet 
(the  bound state occurs on a higher $E$ level). 

In particular, for the Zakharov eq. (\ref{n3}): 
\beq{z5}               
E = - \Omega^{2} /2 + \varepsilon q_{o}^{2}/2
\eeq
the norm integral (\ref{n1}) is elementary: 
\beq{z6} 
N(\psi) =
\int_{-\infty}^{+\infty}q(t)^{2}dt = 
\frac{1}{\sqrt{2}}
\int_{o}^{q_0^2} \frac{d\zeta}{\sqrt{-E +\varepsilon\zeta/2}} =
\frac{4}{\varepsilon\sqrt{2}}
[\sqrt{-E +\varepsilon q_{o}^{2}/2}-\sqrt{-E}]
\eeq
Determining $q_{o}^{2}$ from (\ref{z5}),
\beq{z7}               
q_{o}^{2}  = (2E + \Omega^{2})/\varepsilon
\eeq
and substituting into (\ref{z6}) with $N = 1$ one obtains,
\beq{z8}               
\sqrt{-2E} = \Omega - \varepsilon/2 
\eeq
Since $\sqrt{-2E}$, from definition, is non-negative, the solution  
exists  only if $\varepsilon\leq 2\Omega$ (too strong non-linearity 
prevents the localizations! Compare with an approximate result in the theory of Bose condensation \cite{Sack}).
One henceforth obtains: 
\beq{z10}               
E = -(1/2) (\Omega - \varepsilon/2)^{2}
\eeq
Note also, that if $\Omega < 0$, $\varepsilon \le 4\Omega$ 
(the case of delta barrier and self attractive packet) 
the long  cues of Zakharov packet make possible
the construction of a huge localized state with the eigenvalue given by 
the identical formula (\ref{z10}).

\section{Macrostates}
\setcounter{equation}{0}

     For a class of non-linear eqs.  the `bound states' exist  even in 
absence of potential.
The phenomenon depends obviously on the global structure of 
the squeezing group. It occurs whenever for some $E$ the  1-dim  
manifold $H=0$ is a  closed  loop  (the  expanding  and 
 shrinking  lines $I_{\pm}(E)$
connect). If this is the case, the time dependent point moving along 
the loop $I_{\pm}(E)$ paints a huge localized state $\psi$.  
The  norm  of $\psi$ 
cannot be arbitrarily small (it is exactly  determined  by  the 
value of E and by the nonlinearity); we thus call $\psi$ 
a {\it macro-state}  or 
{\it macro-localization}.  The  links  with  solitons   are immanent. 
The `solitonology' usually describes `traveling  waves' $\Psi(x,t)$,  
with  the stability aspects stressed and spectral aspects forgotten. 
However,  for 
the Galileo invariant theory  both  phenomena  are  the  same:  the 
`traveling waves' are just Galileo transformed macrostates.  This  would 
suggest that the solitons too must share the bifurcation aspects of the  bound 
states! We shall show that this indeed occurs. Let us return to
the non-linear Schr\"{o}dinger's  eqs. (\ref{p2}-3).
\begin{figure}[H]
\includegraphics[width=\linewidth]{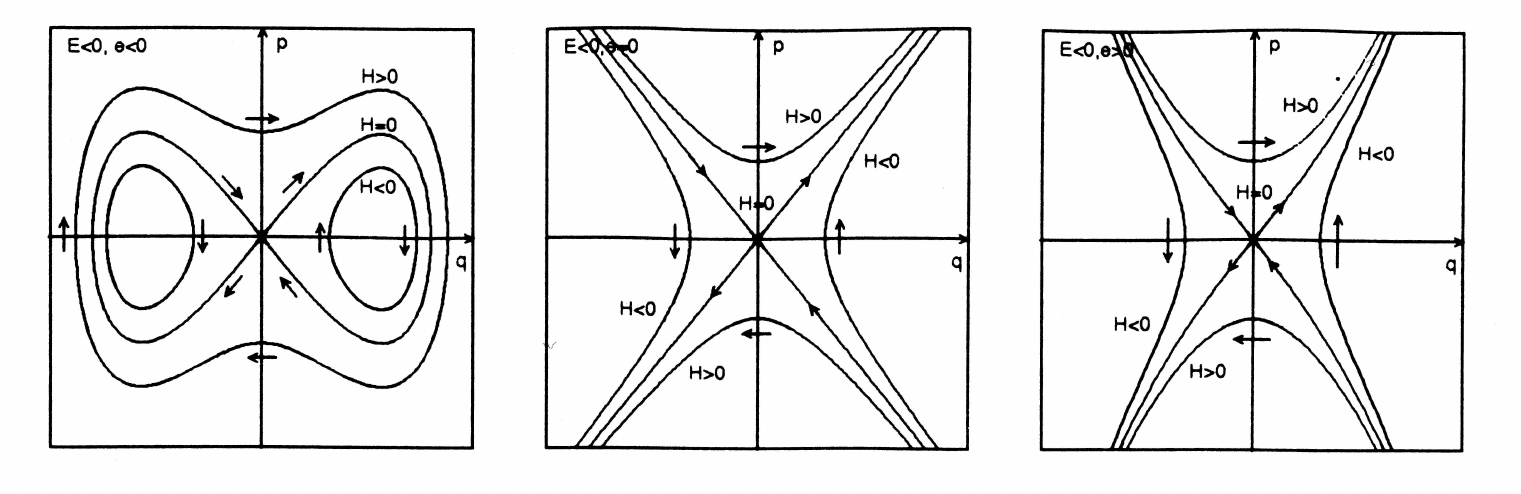}
\caption{
The structure of $H=0$ izolines makes possible the 
`macro-localizations'
for the Zakharov nonlinearity~\eqref{9.1} with $\varepsilon < 0$ and
$E<0$; for the case of Gausson a similar picture 
would be obtained for $\varepsilon < 0$ and $E$ arbitrary.
}\label{fig4}
\end{figure}

The `macro-localizations' of (\ref{p2}) exist  if $F(\zeta)
=\int f(\zeta)d\zeta$ 
in $[0,\infty]$ has a proper but local maximum at $\zeta=0$. 
If this  is  the  case, 
each canonical orbit of (\ref{p3}) emerging from 0 reaches a
maximal $q$-value at the intersection of the $I_{\pm}(E)$ line with 
the $q$-axis,  and 
then returns to 0, drawing a picture of a  macro-state.  
Whether  this 
requires a finite or infinite time depends on the nonlinearity  
function $f()$. If the time is finite, the system has tendencies  
to create compact support solutions (droplets), a phenomenon which 
still awaits investigation (but see an interesting article 
of Aronson, Crandall and Peletier \cite{Aro}).
Below, we study two cases of (\ref{p2}-3) 
for which the time is infinite 
(localizations vanish asymptotically as $t\rightarrow\pm\infty$): 
\begin{align}
\label{9.1}
\text{Classical~soliton~(Zakharov):}~~~f(\zeta)&=\zeta\\
\label{9.2}
\text{Gausson~(Bialynicki-Birula~\&~Mycielski):}~~~f(\zeta)&=\ln \zeta
\end{align}
The topology of the `squeezing lines' in both cases 
corresponds to our Fig.~\ref{fig4}, though the detailed 
behaviour of the cues is different.
As one can see, the Zakharov eq. has the macro-states for 
$\varepsilon<0$ and $E<0$; the logarithmic eq. for $\varepsilon<0$ and 
any $E$. Both integrate easily, leading to the well known 
formulae: 
\begin{align}
\label{9.3}
{\displaystyle \psi_{S}(x)} &=
\sqrt{\frac{2E}{\varepsilon}}\frac{1}{\cosh(x\sqrt{-2E})}
~~~{\rm (Classical~soliton)},\\
\label{9.4}
{\displaystyle \psi_{G}(x)} &=
\exp\left(\frac{E +\varepsilon}{2\varepsilon} 
+ \varepsilon x^{2}\right)
~~~{\rm (Gausson)}.
\end{align}
The links between the solitons and completely integrable mechanical 
systems have been  carefully  explored \cite{Wo1,Wo2},  though  
the  attention  was 
usually focused on the soliton  stability  (with  few 
exceptions; see, e.g. \cite{Bare}). For gaussons (\ref{9.4}) the
stability problem is still open. 
The questions as to, how  the soliton (gausson)
interacts with an external potential was almost neglected: and  
this is precisely
where  the  bifurcations occur.  In  fact,  consider  the  canonical 
trajectory of eqs. (\ref{p3}) which departs from $q=0$ at  $t=-\infty$  
and draws a macrostate. Suppose, however, the process is perturbed
by a little potential pulse $V(t) = \lambda\phi(t)$, where
$\lambda\in R$ and $\phi(t)\not\equiv 0$ is a fixed, bounded,  non-negative 
function vanishing outside of a finite interval $(\alpha,\beta)$.
One might expect that if $\varepsilon$ is small enough, the existence 
of $V(t)$ will cause just a little modification in the form of each macrostate.
In general though, this is not the case. Indeed one can show that even a 
very tiny potential pulse can preclude completely the existence of 
a class of stationary states which exist in vacuum.

 Let $I_{\pm}(E)$ be the macrostate loop for (\ref{s3}-4) and let 
${\bf q}(t)=(q(t),p(t))$ be one of the corresponding localized solutions defined by (\ref{s5}). To fix attention,
choose ${\bf q}(t)$ on the upper branch $I_{+}(E)$, with $p(t) > 0$ 
and $q(t)$ increasing from 0 to $q_{\rm max} = q(t_{o})$, 
(where $q_{\rm max}$ 
is the maximal value of $q$ 
at the turning point between $I_{+}(E)$ and $I_{-}(E)$). 
Suppose, the potential pulse $V(t)=\lambda \phi(t)$ 
occurs in an interval $(\alpha,\beta)$ where ${\bf q}(\alpha), {\bf q}(\beta)\in$
$I_{+}(E)$, $0 < q(\alpha), q(\beta) < q_{\rm max}$, $0 < p(\alpha),p(\beta)$.
 Then:

\begin{prop}\label{prop3}
  If $|\lambda|$ is small enough, $\lambda \neq 0$, then eqs. (\ref{c4}) have no bound state coinciding with ${\bf q}(t)$ for $t \leq \alpha $.
\end{prop}

\begin{proof}
Let $\tilde{{\bf q}}(t)$ be the integral trajectory of (\ref{c4}) for $V(t) = \lambda \phi (t)$ with $\tilde{{\bf q}}(\alpha) = {\bf q}(\alpha)$. If $\lambda$ is small enough , then $\tilde{{\bf q}}(t)$ is arbitrarily close to ${\bf q}(t)$ in $(\alpha,\beta)$ together with its 1-st derivative. In particular, $\tilde{q}(\beta)$ must be close to $q(\beta) \Rightarrow  0 < \tilde{q}(\beta)< q_{\rm max}$. Moreover, for $t\in(\alpha,\beta),$  $\tilde{{\bf q}}(t)$ intersects each vertical line $q=const$ exactly once, suggesting q as a convenient integration variable to compare both trajectories. Consider therefore  the $q$-interval
 [$q_{\rm 1}, q_{\rm 2}$] where $q_{\rm 1}= q(\alpha) = \tilde{q}(\alpha), 
  q_{\rm 2}= \tilde{q}(\beta) $ and  compare $\tilde{p}(q_{\rm 2})$ and $p(q_{\rm 2})$. 
The canonical eqs. (\ref{c4}) imply:
\beq{9.5}
\frac{dp}{dq} = \frac{-2Eq + 2\varepsilon f(q^{2})q}{p},~~~
\frac{d\tilde{p}}{dq} = \frac{-2Eq + 2\varepsilon f(q^{2})q 
+ 2\lambda\phi(t(q))q}{\tilde{p}}
\eeq

\beq{9.6}
\frac{d}{dq}\left(\frac{\tilde{p^{2}}}{2} - \frac{p^{2}}{2}\right)
= 2\lambda\phi(t(q))q \Rightarrow
\frac{\tilde{p}(q_{\rm 2})^{2}}{2} - \frac{p(q_{\rm 2})^{2}}{2} = 2\lambda\int_{q_{\rm 1}}^{q_{\rm 2}}
\phi(t(q))qdq .
\eeq

If now $\lambda > 0$ then ${\tilde{p}}(q_{\rm 2}) > p(q_{\rm 2}) $ and the point
$(q_{\rm 2}, {\tilde p}(q_{\rm 2})) = \tilde{{\bf q}}(\beta)$ ends up in the outer 
region of the loop $I_{\pm}(E).$ To the contrary, if $\lambda < 0 $ and $|\lambda |$ small, then $\tilde{{\bf q}}(\beta)$ falls into the interior of the loop.
In both cases, ${\bf q}(\beta)$ for $t \ge \beta$ originates a periodic trajectory which circulates either in the external or internal region, without ever returning to 0.
\end{proof}

Intuitively, no  matter  the 
value of $E$, there is no localization (macrostate)
with an infinitesimal 
pulse  under one of its cues. The presence of $V(t)$, (no  matter  how 
tiny)  
must displace the orbit out of  the  `squeezing  lemniscate' $I=I_{\pm}(E)$, 
 creating an unlimited  trajectory and a class of macrostates 
disappears, (a  `delicate condition' of macrostates due to the fact that they exist on the  
threshold of the symmetry breaking. Note that, our proposition adds only 
some details to
the sequence of theorems on the perturbations of homoclinic
orbits (see \cite[Sec.~4.5]{Holms} and the literature given there). 
The recent results in the theory of Bose condensation \cite{Boson1,Boson2,Sack,Dum,Bare} are interpretable as an indirect consequence of the same mathematical theorems \cite{Holms}.

A macrostate can survive the perturbation if the `little obstacle' is  under  its  `mass 
center'. Geometrically, it means that   
the pulse $V(t)$  just  modifies  the outer part of $I_{\pm}(E)$ 
in vicinity of $q_{\rm max}$. It can also create
a bi-localized state by establishing a bridge between two upper or two lower
branches of $I_{\pm}(E)$ (see Fig.~\ref{fig5}). As a result, the  entire  
`macrostate  park' is different in presence of  an  arbitrarily  
small  $V(t)$.  

      The metamorphosis is even more radical in presence of two little 
and widely separated potential pulses, $V(t)  = V_{1}(t) + V_{2}(t)$,
 too small  to  have  any  bound  states in the traditional sense.  Indeed, 
consider again a classical trajectory which 
starts to `paint a macrostate' to the left of the left pulse $V_{1}$.  
Then it may happen that $V_{1}$ deflects the `classical point'
(\ref{s2}-3) from the 
lemniscate $I_{\pm}(E)$, to the internal (or  external) loop $H$=const, 
where it starts to circulate, but the second pulse $V_{2}$, turns 
it back to $I_{\pm}(E)$ where it is finaly driven to 0.
The resulting macrostate  has 
the form of a multi-localization, a stationary state which could
not exist at all in the collection of vacuum states. Now, it  is 
created just by two tiny `potential clips' (see  Fig.~\ref{fig5}).
Since the stationary states are natural reference points for the general 
motion, one might expect that a bifurcation  in the sets of macrostates
must also affect the evolution of nonstationary packets.
We shall see below that this is indeed the case.
\begin{figure}[H]
\centering
\includegraphics[width=0.72\linewidth]{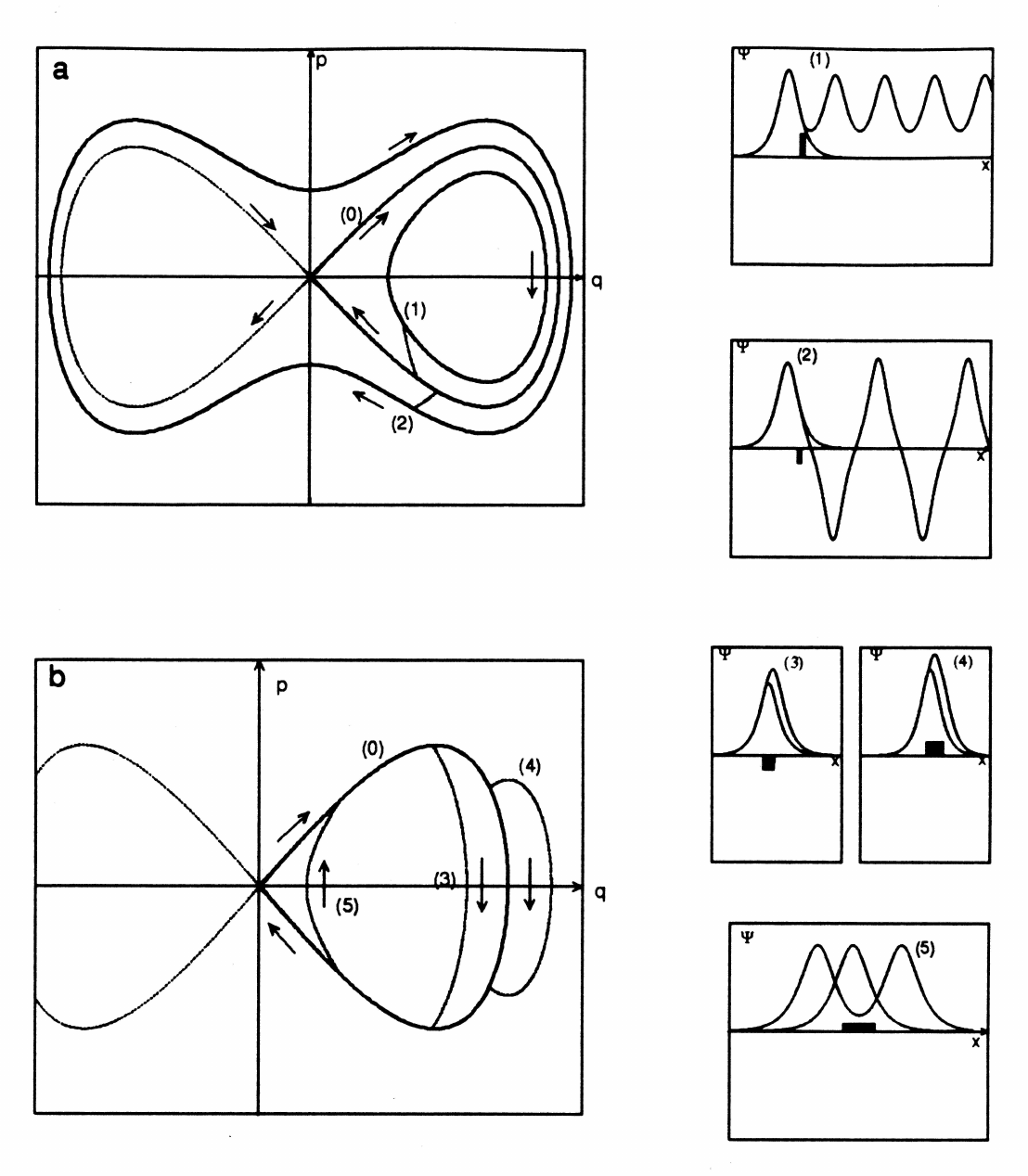}
\caption{
The dynamical symmetry breaking caused by an arbitrarily
tiny pulse $V(t)$ makes impossible the existence
of a macrostate with $V(t)$ under one of the cues: (1) a little
positive $V(t)$ (potential barrier) applied to the canonical
trajectory departing from the origin at $t = -\infty$ kicks
the trajectory into the internal closed loop $H$ = const $<$ 0,
where it starts to circulate endlessly without returning to the
origin. (2) a negative $V(t)$ (potential well) pushes.
the trajectory toward an external lemniscate $H$ = const $>$ 0,
where it again circulates, without returning to the origin.
(3-4) The macrostates with $V(t)$-pulse collocated under the center
of the wave packet are possible either for the well
or barrier. (5) A new macrostate in form of a
bi-localization can be also created by arbitrarily tiny potential 
pulse.}
\label{fig5}
\end{figure}     
\begin{figure}[H]
\centering
\includegraphics[width=0.9\linewidth]{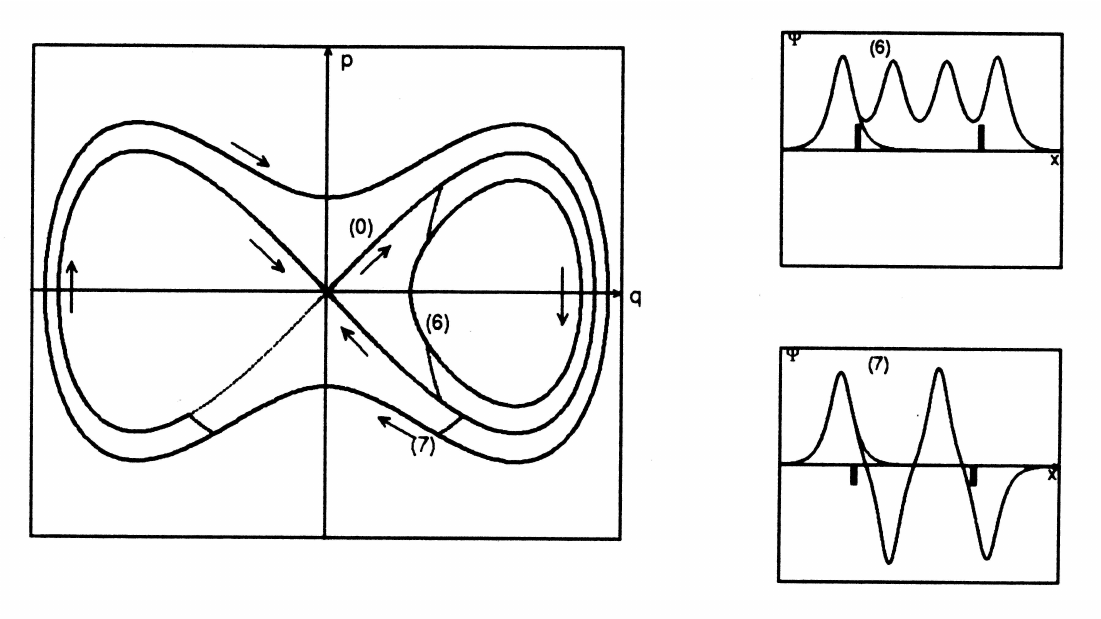}
\caption{
A pair of arbitrary weak barriers (wells) permits the 
multi-soliton states: (6) a multisoliton created by several circulations
of the classical trajectory on the internal $H$-curve, `clipsed' by two
tiny potential barriers, (7) an analogous multi-soliton created by
an external circulation, is maintained by two tiny wells.}
\label{fig6}
\end{figure}     

\section{The non-stationary states.}
\setcounter{equation}{0}

The stationary solutions have a guiding role in linear theories.
Would the same be true in the non-linear case? 
To check this we  have  returned  to  
(\ref{p1}) and applied numerical techniques to examine some 
general, non-stationary waves $\Psi(x,t)$ (with $t$ meaning 
again the time and $x$ the space coordinate). We  were  specially  
curious to see what happens to the initial macrostate (\ref{9.3}) 
or (\ref{9.4}) in presence of a very little potential pulse $V(x)$:
\beq{10.1}
V(x) = \left\{
\begin{array}{l}{\displaystyle 
V_{o}\left[1 - \frac{(x-x_{v})^{2}}{\sigma^{2}}\right]^{2},
~~~|x-x_{v}|\le\sigma}\\[1ex]
0,~~~|x-x_{v}|>\sigma
\end{array}\right.
\eeq
We have first taken a small $V_{o} < 0$ and $x_{v} < 0$, to
represent a little 
potential well situated under the left vanishing cue of the initial 
macrostate, and we have examined the evolution of the packet 
$\Psi(x,t)$ in the initial form of Gausson for the non-linearity (\ref{9.2}). 
The result is curious: the packet $\Psi(x,t)$ is  first  attracted 
toward the well, then starts to perform around it a sequence of 
decaying oscillations loosing an `excess of matter' and tending 
slowly to a new equilibrium state (of smaller norm), right on 
the center of the well;  see  Fig.~\ref{fig7}a. Its asymptotic 
form is therefore affected by an arbitrarily small $|V_{o}|$.
The similar effect is observed for the
Zakharov soliton, see Fig.~\ref{fig7}b-e. Evidently, while there is no 
discontinuity of the finite time evolution of the
soliton (Gausson) due to  
the  influence of $V(t)$ there is a discontinuity (bifurcation)
in its asymptotic 
form (meaning that the limiting transitions $V\rightarrow 0$ and 
$t\rightarrow +\infty$ do not commute). 

Our next experiment involved the Gausson initially situated almost
upon the center of a little potential barrier. As turns out the 
simulation can generate several scenarios, the most interesting
one is the Gausson splitting illustrated on Fig.~\ref{fig8}c. 
(In our computer simulations we also observed an analogous phenomenon 
for the traditional Zakharov soliton). The splitting occurs as well 
for the traveling soliton of the initial form,
\beq{eq:8.2}
\Psi(x,0) = \psi_{S}(x)\exp(ikx),
\eeq
colliding with the potential barrier (where $\psi_{S}$ is given 
by \ref{9.3}). As turns out,
the too slow soliton is totally reflected and too quick
soliton is totally transmitted, Fig.~\ref{fig9}; the splitting
occurs for intermediate packet velocities. We conclude that
the solitons, while stable in mutual collisions, can loose stability 
in presence of external potentials.

Our last experiment was to examine the influence of two tiny
potential wells placed symmetrically under two cues of the
initial Zakharov soliton. Our calculations show that the
bi-soliton stationary state (Fig.~\ref{fig6}) dictates the
asymptotic form of the non-stationary wave $\Psi(x,t)$, see
Fig.~\ref{fig10}. We thus see that 
while a tiny pulse of the external potential cannot affect the
continuity of the macrostate evolution, 
it can however cause the bifurcation
of their asymptotic forms consistently with our idea of the
non-linear spectral phenomenon. It is interesting to notice, 
that the bifurcations 
which we are describing arised already in some applied areas. 
Thus, the metamorphosis of the localized stationary state into 
a `respiring lump' 
resembles phenomenon predicted in physics of boson condensation 
(see \cite{Sack}); likewise the split soliton seems an abstract 
equivalent of an effect discused in \cite{Dum}, both belonging 
to the newly emerging side of the soliton theory.

\begin{figure}[H]
\centering
\includegraphics[width=\linewidth]{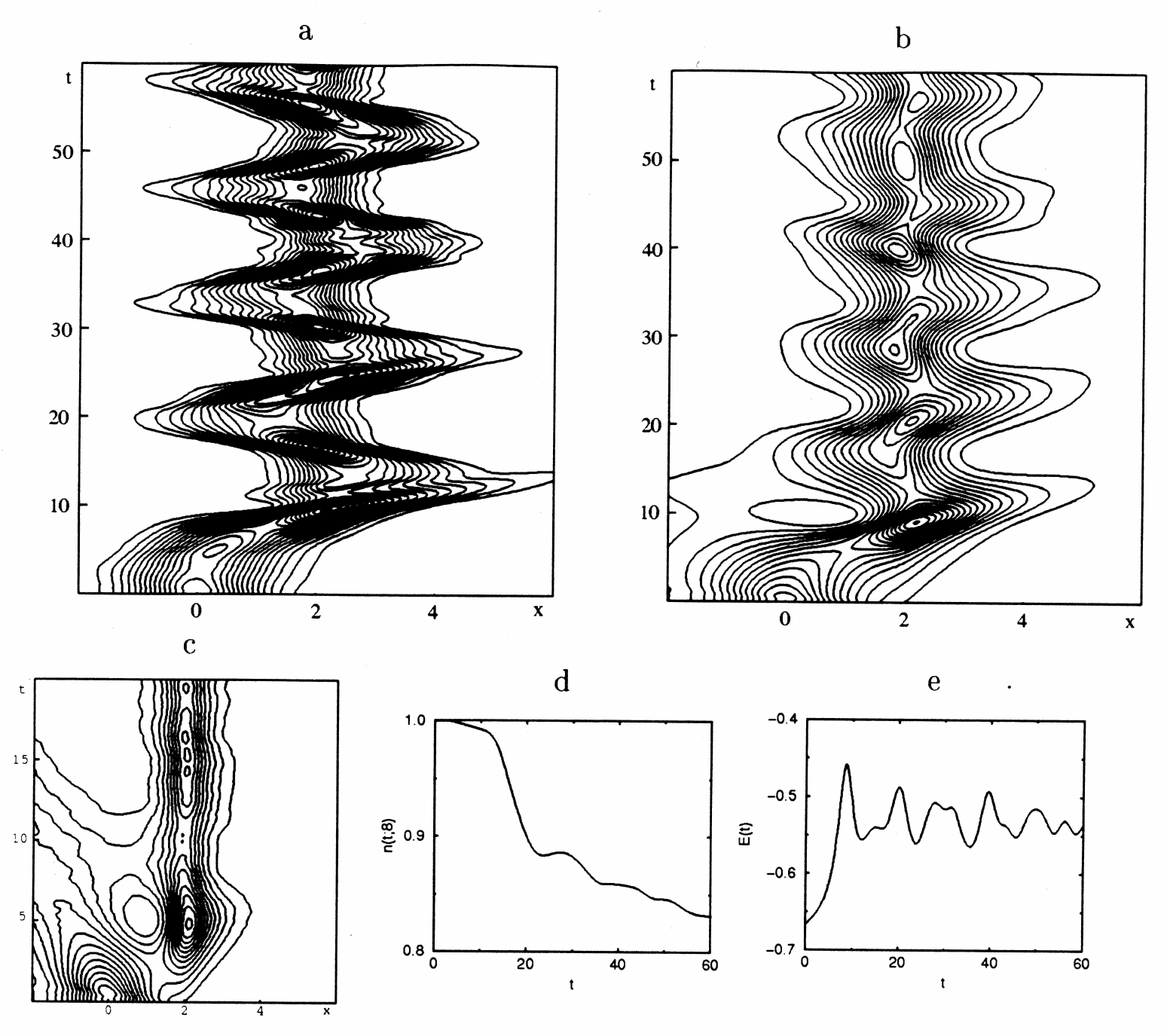}
\caption{
Effects of a little well~\eqref{10.1} placed under the cue 
of `macrostate' $\Psi(x,0)$. The evolving packet represented by 
the izolines of $|\Psi(x,t)|^{2}$ tries to find an
equilibrium state upon the center of the well (eq., 
$x_{v}$ = 2, $\sigma$ = 0.5, $V_{o}$ = -0.75): (a) $\Psi(x,0)$ is a 
Gausson~\eqref{9.4}; (b) $\Psi(x,0)$ represents the 
soliton in Zakharov equation~\eqref{9.3}; (c) a detail of the Zakharov 
process (for $V_{o}$ = -2) shows probable emission of a part of the 
soliton substance when tending to an equilibrium state; 
(d) the time evolution of the partial norm $n_{[-8,10]}(t)$ 
eq.~\eqref{m3} confirms the previous conclusion; (e) the generalized 
frequency parameter $E$ given by~\eqref{m4} tends to a new value 
for the new equilibrium state.}
\label{fig7}
\end{figure}
\begin{figure}[H]
\includegraphics[width=\linewidth]{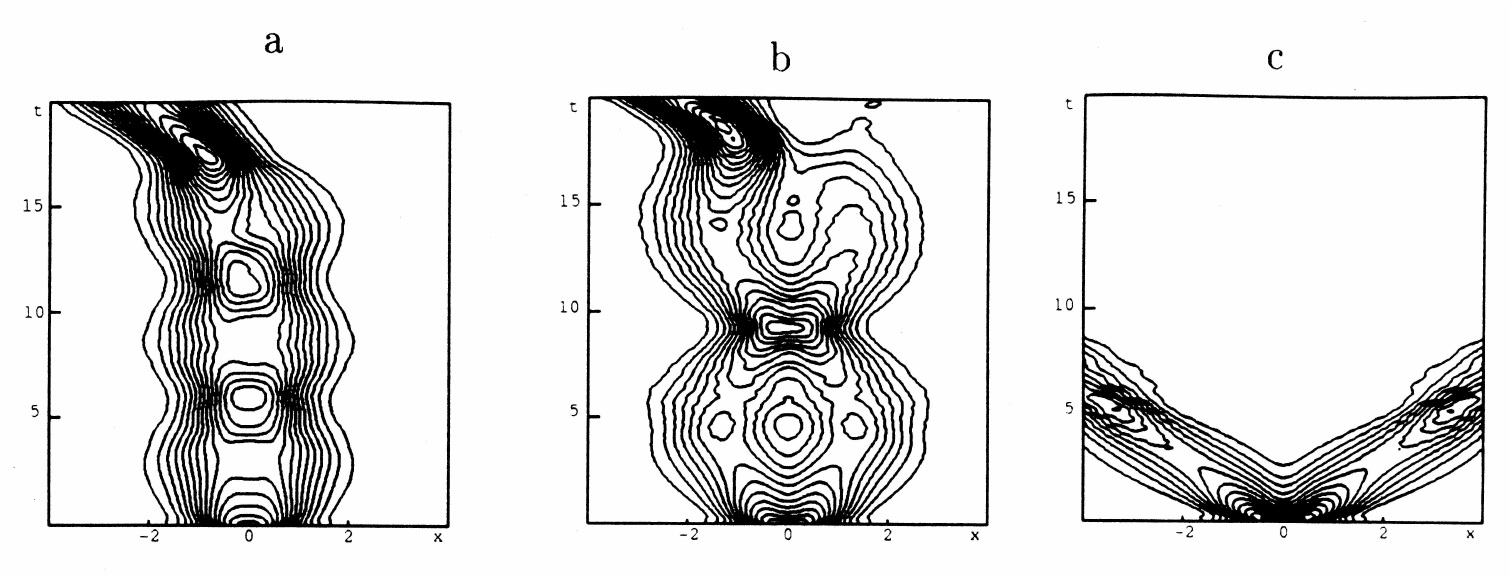}
\caption{
Cases of Gausson evolution in presence of the
barrier~\eqref{10.1} centered almost at the maximum of the
initial $\Psi(x,0)$ ($x_{v}$ = 0.001, $\sigma$ = 0.5).
(a) For  low $V_{o}$ = 0.5 the Gausson hesitates but than deflects to
the left without loosing integrity, (b) an analogue phenomenon
for $V_{o}$ = 1 suggests an emission of the Gausson substance,
(c) the barrier $V_{o}$ = 2 causes the new phenomenon of
Gausson splitting.}
\label{fig8}
\end{figure}
\begin{figure}[H]
\includegraphics[width=\linewidth]{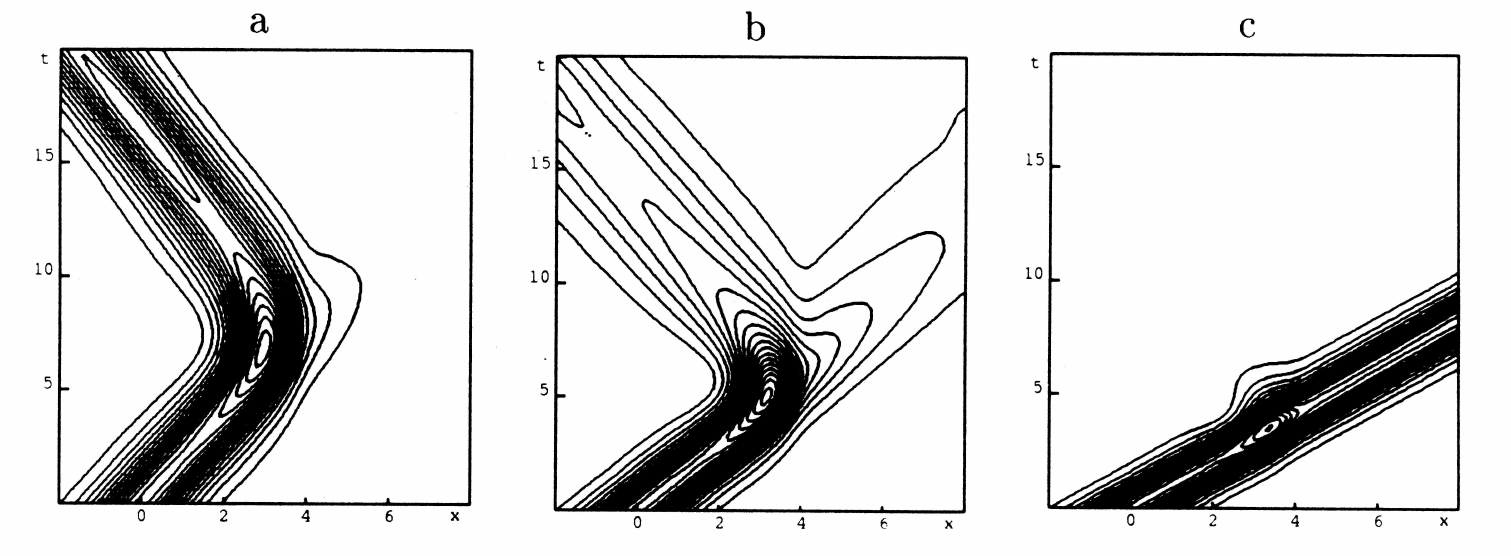}
\caption{
The evolution of the Zakharov soliton $\Psi(x,0) = \psi_{Z}(x)\exp(ikx)$
colliding with the potential barrier of form~\eqref{10.1} with $x_{v}$=4,
$\sigma$ = 0.5: (a) for $k$ = 0.5, $V_{o}$ = 1 the soliton is totally
reflected; (b) for $k$ = 0.7, $V_{o}$ = 1 the soliton splits; (c) for
a slightly faster soliton $k$ = 1 and lower barrier $V_{o}$ = 0.5 the
packet is totally transmitted.}
\label{fig9}
\end{figure}
\begin{figure}[H]
\includegraphics[width=\linewidth]{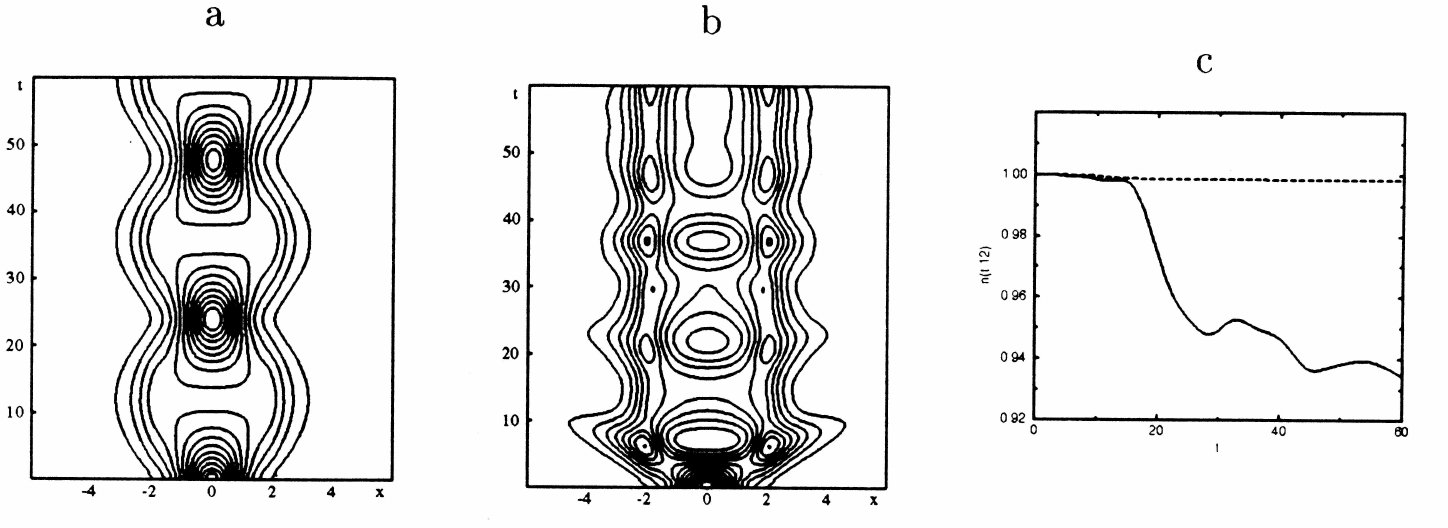}
\caption{
The evolution of the Zakharov `macrostate'~\eqref{eq:8.2} in presence of
two potential wells of form (eq. , $\sigma$ = 0.5) placed at 
$x = \pm 2$: (a) for $V_{o}$ = -0.5 the packet performs
quadrupole oscillations around
wells; (b) for $V_{o}$ = -1 it appears to tend to the stationary
bi-soliton state; (c) the partial norms of the soliton $n_{[-12,12]}(t)$
for the both processes suggests that the new equilibrium in the case (b)
is achieved at the cost of emitting an excess of substance.}
\label{fig10}
\end{figure}

\section*{Acknowledgments}

The authors are grateful to their colleagues at ICM and
Institute of Theoretical Physics, Warsaw University, Warsaw,
Poland, and in Departamento de Fisica, CINVESTAV, Mexico, for
their interest in the subject and pertinent discussions.
One of us (BM) is grateful to professors C.V.~Stanojevic and
W.O.~Bray for their
kind invitation to the VI-th IWAA, Maine, US, June 1997,
where a part of this work has been presented. Two of us (PG and WK) 
are indebted for the kind invitation to the Departamento de F{\'\i}sica, 
CINVESTAV, Mexico. The work was partially supported by the Polish
State Committee for Scientific Research . 
The essential contribution of the computer facilities of ICM is 
acknowledged.

\section*{Appendix: Numerical Algorithms}
\newcommand{\iniappen}
{\renewcommand{\theequation}{A.\arabic{equation}}}
\iniappen
\setcounter{equation}{0}
For numerical solution of the equation (\ref{p1}) we apply  
a very simple discretization of the
space and time preserving the fundamental symplectic character
of the dynamics. The equation is represented in form 
of the classical Hamiltonian equations for a continuous medium,
\beq{m1}
\pder{Q}{t} = \frac{\delta {\cal H}}{\delta P},~~~
\pder{P}{t} = -\frac{\delta {\cal H}}{\delta Q}.
\eeq
The real numbered functions $Q(x,t)$ and $P(x,t)$ represent
real and imaginary part of $\Psi(x,t)$ and ${\cal H}$ is the 
following functional,
{\beq{m2}
{\cal H}\left[Q,P\right] 
= \frac{1}{2}\int\left[-\frac{Q}{2}\frac{d^{2}Q}{dx^{2}} 
- \frac{P}{2}\frac{d^{2}P}{dx^{2}} 
+ (Q^{2} + P^{2})V + \varepsilon F(Q^{2}+P^{2})\right]dx.
\eeq}
Note that values of ${\cal H}$
and of the norm of the wave function are preserved in the
evolution. The functions $Q$ and $P$ are represented
on the finite regular grid of points in the domain of $x$ and
the Laplacian is approximated with a finite difference formula.
Consistently, the equations (\ref{m1}) are transformed into 
ordinary Hamiltonian equations for many degrees of freedom. 
In practice we apply the grid $x\in[-16,16]$ with 801 points. An
appropriately extended grid is required for simulations of the 
traveling packets. In all cases where the emission of the 
`packet substance' is reported we additionally applied the
technique of absorbing boundaries on the edges of the grid.
This representation does not produce
any substantial artifacts of the space discretization. For the
time integration we use the implicit second order Range-Kutta 
method which is strictly symplectic and appropriate for
classical Hamiltonians which are not separable into a kinetic 
and potential part \cite{Sanz94}. The integration time step 
is $2\times10^{-4}$ ensuring absolute stability and elimination 
of any substantial artifacts of the time discretization. In order
to analyze the evolution of the packet we introduce the norm
$N$ and the partial norm $n_{L}$:
\beq{m3}
n_{L}(t) = \frac{1}{N}\int_{L}|\Psi(x,t)|^{2}dx,~~~
N = \int|\Psi(x,t)|^{2}dx, 
\eeq
and the field functional $E(t)$ generalizing the eigenvalue $E$
of (3.2) to the nonstationary solutions,
\beq{m4}
E(t) = \frac{1}{N} \int\Psi^{\ast}(x,t)\left[
-\frac{1}{2}\pder{^{2}}{x^{2}} + V(x) + 
\varepsilon f(|\Psi(x,t)|^{2})\right]\Psi(x,t)dx.
\eeq
The simulations have been performed with help of the
fortran program with 48-bits representation of real 
numbers running on Cray Y-MP/4E.

\end{document}